\documentclass[conference]{IEEEtran}
\IEEEoverridecommandlockouts
\usepackage{cite}
\ifCLASSINFOpdf
\fi
\usepackage{amsmath,amsfonts,amsthm} 
\usepackage{amssymb}
\usepackage{bbm}
\usepackage{float}
\usepackage{graphicx}
\usepackage{algorithm}
\usepackage{algpseudocode}
\usepackage{siunitx}
\usepackage{tikz}
\usetikzlibrary{arrows,automata}
\usepackage{multicol}
\usepackage{comment}
\usepackage{pgfplots}
\usepackage{epstopdf}
\usepackage{multirow}
\usepackage{balance}
\usepackage[font=scriptsize,skip=2pt]{caption}
\usepackage{soul}
\usepackage{todonotes}
\usepackage{color, colortbl}
\usepackage{array}
\usetikzlibrary{patterns}
\usepackage{balance}

\usepackage{enumitem}
\usepackage{footmisc}
\usepackage{pgfplots}
\usepackage{subcaption}
\usepackage{blkarray}
\usepackage{tabulary}
\usepackage{booktabs}
\usepackage{stfloats}
\newtheorem{theorem}{Theorem}

\makeatletter
	{\end{proof}%
	\renewcommand{\qedsymbol}{\qedsymbol}}
\makeatother

\usepackage{cuted}
\usepackage{tikz}
\usetikzlibrary{automata, positioning}

\allowdisplaybreaks
\usepackage[implicit=false]{hyperref}

\usetikzlibrary{shapes.geometric}
\usetikzlibrary{decorations.pathreplacing,decorations.markings,calc}

\graphicspath{{images/}}

\DeclareMathAlphabet\mathbfcal{OMS}{cmsy}{b}{n}

\definecolor{Gray}{gray}{0.9}
\definecolor{LightCyan}{rgb}{0.88,1,1}

\algnewcommand{\IIf}[1]{\State\algorithmicif\ #1\ \algorithmicthen}
\algnewcommand{\IELSEIF}[2]{\State\algorithmicelseif \ #2\ \algorithmicthen}
\algnewcommand{\IElse}[3]{\State\algorithmicelse\ #3\ }
\algnewcommand{\EndIIf}{\unskip\ \algorithmicend\ \algorithmicif}

\IEEEoverridecommandlockouts

\tikzstyle{Block} = [rectangle, rounded corners, minimum width=3cm, minimum height=1cm,text centered, draw=black, fill=red!30]

\usetikzlibrary{shapes.geometric}
\usetikzlibrary{decorations.pathreplacing,decorations.markings,calc}

\definecolor{Burgundy}{RGB}{10,0,0}

\usetikzlibrary{positioning}
\usetikzlibrary{arrows,decorations.markings}
\usetikzlibrary{arrows}

\usepackage{algorithm,algcompatible,amsmath}
\algnewcommand\INPUT{\item[\textbf{Input:}]}%
\algnewcommand\OUTPUT{\item[\textbf{Output:}]}%

\begin{document}	

\title{Distortion Minimization with Age of Information and Cost Constraints
}
\author{\IEEEauthorblockN{Jayanth S\IEEEauthorrefmark{1}, Nikolaos Pappas\IEEEauthorrefmark{2} and Rajshekhar V Bhat\IEEEauthorrefmark{3}}
    \IEEEauthorblockA{\IEEEauthorrefmark{1}\IEEEauthorrefmark{3}Department of Electrical Engineering, Indian Institute of Technology Dharwad, Dharwad, India\\
	\IEEEauthorrefmark{2}Department of Computer and Information Science, Link{\"o}ping University, Link{\"o}ping, Sweden\\
		Emails: 
		\IEEEauthorrefmark{1}jayanth.s.20@iitdh.ac.in, 
		\IEEEauthorrefmark{2}nikolaos.pappas@liu.se, 
        \IEEEauthorrefmark{3}rajshekhar.bhat@iitdh.ac.in. 
}
\thanks{The work of N. Pappas has been supported in part by the Swedish Research Council (VR), ELLIIT, Zenith, and the European Union (ETHER, 101096526). 
The work of Rajshekhar V Bhat has been supported by the Science and Engineering Research
Board, Department of Science and Technology, Government of India, under
Project SRG/2020/001545.
}}

\maketitle

\begin{abstract}

We consider a source node deployed in a real-time monitoring application that needs to sample a stochastic process and \emph{convey its state} timely and accurately to a destination over a wireless ON/OFF channel. The source can either process a raw sample to determine its current state and transmit that information or transmit the raw sample and let the destination determine the state. The source is subjected to an average cost constraint, and it cannot sample, process, and transmit at all the time instants due to the associated costs. 
When the destination does not receive information, it uses the previous information as an estimate of the current state, which, if it matches the actual state at the source, the distortion is considered to be zero. The objective is to minimize average expected distortion subject to constraints on the average expected age of information (AoI) of states of interest and costs incurred by the source, where the AoI of a state increases if no status update is received, else drops to unity. We derive a stationary randomized policy (SRP) to solve the formulated problem, for which we obtain the expression for the expected AoI under the SRP using a lumpability argument on the two-dimensional discrete-time Markov chain formed using AoI and instantaneous distortion as states. We extensively study the impact of the system parameters on the average distortion under the SRP and draw significant conclusions. 
\end{abstract}

\section{Introduction}

Next-generation communication systems require fresh delivery of information. Age of Information (AoI) as a metric of the freshness of information was introduced in \cite{aoi-sanjit}. 
Since then, it has received significant attention from the research community \cite{Book_Nikos, Book_Yates, Pappas2023age} and has led to the introduction and optimization of several other related metrics, such as the non-linear AoI \cite{KostaTCOM2020, KostaJSAC2021}, the Age of Incorrect Information \cite{AoII, AoIE}, Age of Synchronization \cite{AoS}, version AoI \cite{VAoI_Yates, VAoI_Ulukus}, and Age of Actuation \cite{AoA2023}. Minimizing AoI requires joint optimization of sampling intervals, queuing and transmission delays, allocation of resources such as transmit power and bandwidth, and also the processing delay in case the information needs to be processed for being useful \cite{Book_Nikos, Preprocess-Tony, AoP, AoPI}.

Queuing and transmission delays, which affect the AoI, may be reduced by compression or packet management techniques or by adopting high transmit power. In the former case, one may incur distortion or processing delay, leading to a trade-off between AoI, distortion, and transmission power. Such trade-offs appear in systems requiring fresh information delivery, as studied in \cite{age-distortion-RVB, QoI_Rahul, Preprocess-Tony}.   
Another approach used for a fresh delivery of information is by deciding not to transmit the information to the destination if the information available at the destination is in sync with the information at the source and has been considered in \cite{AoII, AoIE} and \cite{NikosICAS21} by comparing the state of the stochastic process at the source and an estimate of the state of the process at the destination. These works assume that the source is aware of its state, without requiring any processing, and that of the estimate at the destination due to the acknowledgment it receives from the destination after successful transmission. 
However, in real-world systems, such as video surveillance, one may need to process the sampled information to know the state of the stochastic process, and has not been accounted for in the earlier works, including \cite{AoII, AoIE} and \cite{Preprocess-Tony}. 
In this work, we aim to decide whether to process a sample at the transmitter or not. 
If processing is carried out at the transmitter, we seek to exploit it to know the state of the stochastic process being monitored and then decide whether to transmit or not. The processing also helps reduce the transmission cost incurred for successfully delivering packets.

\begin{figure*}[t]
	\centering
	\resizebox{0.98\textwidth}{!}{%
		\includegraphics{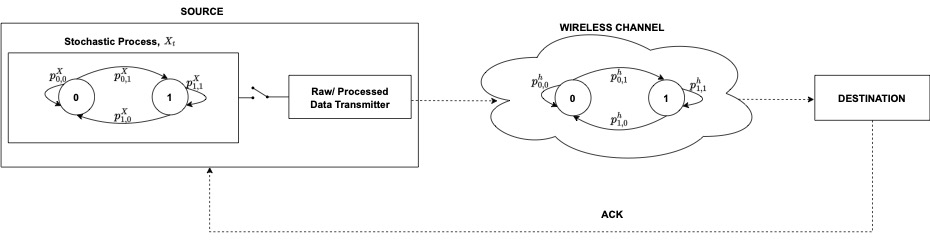}
	}	
	\caption{The source samples the stochastic process and transmits the data, with or without processing, to the destination through a wireless channel. In return, the source receives an acknowledgment corresponding to the transmitted data.}
	\label{fig:system-model}
\end{figure*}

In the present work, we consider a source and destination pair as shown in Fig. \ref{fig:system-model}. The source monitors a stochastic process, and the destination requires timely information on the process at the source to perform a specific task. The information from the source is transmitted to the destination through a wireless channel, and the source receives an acknowledgment (ACK) back from the destination. In order to transmit timely information, the source has to sample the process, and we assume that the sample must be processed to identify the current state of the stochastic process. For instance, in applications such as autonomous driving, forest surveillance, and factory monitoring, sampled data must be processed to detect obstacles, monitor forest fires, and detect anomalies.
The processing can be carried out at the source and transmitted to the destination. However, the source must pay a specific cost for data processing. The other possibility is to transmit the raw information, which is then processed at the destination. The trade-off to be considered while making the decision is the following: 
If the source processes a sample incurring a processing cost, and the current state is in sync with the estimate at the destination, the sample need not be transmitted, thereby avoiding the transmission cost. 
If they are not in sync,  the source needs to transmit the sample, and since the sample is already processed, the transmission cost is reduced. However, due to the processing cost, if the state of the stochastic process is out of sync with the estimate at the destination, processing the sample might not have been helpful in hindsight. 
Towards incorporating and addressing the above trade-off for fresh delivery of information, the main contributions of the current work are:
\begin{itemize}
    \item We formulate a problem for the long-term average expected distortion minimization subject to average expected AoI and cost constraints, where instantaneous cost includes sampling, processing, and transmission costs.     
    \item We propose a stationary randomized policy (SRP) that decides when to sample the process and whether to transmit raw or processed information to the destination. We derive expressions for the expected distortion and AoI under the SRP. We show that a two-dimensional discrete-time Markov Chain (DTMC),  which considers instantaneous AoI and distortion as the states, is lumpable, due to which, the DTMC can be reduced to a one-dimensional DTMC with the AoI as the state, simplifying the derivation of the expression for the expected AoI. 
\end{itemize}

\section{System Overview}
We consider a source monitoring a stochastic process and a destination seeking timely and accurate information on the state of the stochastic process. 
The source is required to sample the stochastic process and convey its state to the destination. The sampled data of the stochastic process has to be processed to identify the state of the stochastic process, which can be performed either at the source before transmission or at the destination after reception.\footnote{The cost of processing at the destination is considered irrelevant.}
We assume that the number of bits required to represent the raw information sample is always higher than that for the processed information.
In this section, we present the source state model, channel model, actions and the costs incurred, AoI and distortion models, constraints, and problem formulation.

\subsection{Source State Model}
We consider a time-slotted system with slot index $t \in \{1,2, \dots\}$. Let $X_t$ represent the \emph{state} of the stochastic process at the source at time instant $t$ and $\hat{X}_t$ represent the estimate that the destination has about the state of the stochastic process, $X_t$, where $\hat{X}_t = X_{\tau}$ and $\tau$ is the latest time instant at which a sample from the source was received at the destination.
The state transition probability $\mathbb{P}(X_{t+1} = x'|X_{t} = x)$ represents the probability of transition from state $x$ to $x'$. We consider that  $X_t \in \{0,1\}$.\footnote{As previously stated, the raw sample of the stochastic process from which its state, $X_t$, is extracted may require more bits to represent. For instance, in a surveillance application, the raw sample may be an image requiring multiple bits to represent, and after processing, the normality or criticality of the state of the environment under surveillance may be revealed, which can be represented only in two bits.} 
We denote $\mathbb{P}(X_{t+1} = 0|X_{t} = 1) = p_{0,1}^X$, $\mathbb{P}(X_{t+1} = 1|X_{t} = 0) = p_{1,0}^X$, $\mathbb{P}(X_{t+1} = 0|X_{t} = 0) = p_{0,0}^X$ and $\mathbb{P}(X_{t+1} = 1|X_{t} = 1) = p_{1,1}^X$.

\subsection{Channel Model}
We consider a wireless channel between the source and the destination. 
We assume that the channel state $h_t$ in slot $t$ can either be in a good or a bad state represented by state $1$ or state $0$, respectively. Successful data transmission depends on the channel state in a slot and whether raw or processed information is transmitted through the channel. When the channel is in state $i$, let $p_i^r$ be the probability of successfully transmitting raw information and let $p_i^p$ be the probability of successfully transmitting the processed information. When the channel is in the bad state, if the raw information, having more number of bits, is transmitted, it will not be delivered to the destination, i.e., $p_0^r = 0$, whereas with non-zero probability the processed information is received successfully when the channel is in the bad state, i.e., $0<p_0^p<1$. When the channel is in the good state, we assume that the processed information is successfully received every time, i.e., $p_1^p=1$, whereas the raw information is successfully received with non-zero probability, i.e., $0 < p_1^r \leq 1$.
We consider that the destination transmits an acknowledgment (ACK) to the source, which helps the source to know the estimate, $\hat{X}_t$, that the destination has about the state, $X_t$ of the stochastic process.


\subsection{Actions and Costs Incurred}
The source can choose one of the following actions in  a slot: (i) do not sample the stochastic process $X_t$, (ii) sample and transmit the raw information, (iii) sample, process, and decide not to transmit the information because the estimate matches the state of the stochastic process, and (iv) sample, process, and transmit the processed information. The actions (i), (ii), (iii) and (iv) incur $c_1$, $c_2$, $c_3$ and $c_4$ units as costs, respectively. We define the following indicator functions:
\begin{flalign*} 
\begin{aligned}
	& I_0(t) =
	\begin{cases}
		1, & \text{if source does not sample in slot $t$,}\\ 
		0, & \text{otherwise,}\\
	\end{cases}
\end{aligned}
\end{flalign*}
\begin{flalign*} 
\begin{aligned}
	& I_1(t) =
	\begin{cases}
		1, & \text{if the source samples and transmits }\\
               & \text{the raw data in slot $t$,}\\
		0, & \text{otherwise,}\\
	\end{cases}
\end{aligned}
\end{flalign*}
\begin{flalign*} 
\begin{aligned}
	& I_2(t) =
	\begin{cases}
            1, & \text{if the source samples, processes, and  } \\
              & \text{does not transmit the data in slot $t$,}\\
		0, & \text{otherwise,}\\
	\end{cases}
\end{aligned}
\end{flalign*}
and 
\begin{flalign*} 
\begin{aligned}
	& I_3(t) =
	\begin{cases}
		1, & \text{if the source samples, processes, and }\\
		& \text{transmits the processed data in slot $t$,}\\
		0, & \text{otherwise.}\\
	\end{cases}
\end{aligned}
\end{flalign*}


\subsection{Age of Information (AoI) Model}
In certain scenarios, it may be necessary to update information on the stochastic process more frequently when it is in certain states than in others. This must be accounted for in making sampling and transmission decisions. In practice, such a situation can occur if one state is more critical than the other\cite{control_alarm}. Hence, we consider AoI only for the state $1$ at the destination, denoted by $a^d_{1}(t)$ in slot $t$. It evolves as follows: 
\begin{equation}\label{eq:age_evolution}
	a^d_{i}(t+1) =
	\begin{cases}
    	0, & \text{with probability } p_1^r \text{ when} \\
     & I_1(t)=1, h_t = 1 \text{ and } \hat{X}_t = 0, \\
    	 0,   & \text{with probability } p_0^p \text{ when } \\
         & I_3(t)=1, h_t = 0\text{ and } \hat{X}_t = 0, \\
          0, & \text{with probability } 1  \text{ when } \\
               & I_3(t)=1, h_t = 1 \text{ and } \hat{X}_t = 0, \\
     
        1, & \text{with probability } p_1^r \text{ when } \\   
        & I_1(t)=1, h_t = 1 \text{ and } \hat{X}_t = 1,\\
        1,  & \text{with probability } p_0^p \text{ when } \\ 
           & I_3(t)=1, h_t = 0\text{ and } \hat{X}_t = 1, \\
        1, & \text{with probability } 1 \text{ when }   \\ 
           & I_3(t)=1, h_t = 1 \text{ and } \hat{X}_t = 1, \\
    	a^d_{i}(t) + 1, & {\rm otherwise}.\\
	\end{cases}
\end{equation}
The source can keep track of the AoI at the destination for state $1$, based on the acknowledgment it receives from the destination after successfully delivering the raw or processed data in a slot. In words, \eqref{eq:age_evolution} indicates that the age of state $1$ at the destination drops to $1$ if we sample, process, and transmit the processed data, and the destination successfully receives and decodes the state as state $1$, or when we sample the data, transmit it to the destination and the destination successfully receives, processes and decodes the state as state $1$. Also, the age of state $1$ is irrelevant if the destination decodes the state as state $0$, i.e., when $\hat{X}_t=0$. To account for this, we drop the instantaneous age of state $1$ to $0$ when $\hat{X}_t=0$. 


\subsection{Distortion Model}
The main objective of the current work is to maintain the information at the destination to be as fresh and accurate as possible. To quantify the accuracy of the estimate $\hat{X}_t$ at the destination regarding the stochastic process $X_t$, we define and consider the following distortion function $\Delta(\cdot,\cdot)$: 
\begin{equation}\label{eq:def-distortion}
    \Delta(X_t, \hat{X}_t) =
    \begin{cases}
    0, & \text{if } X_t = \hat{X}_t, \\
    1, & \text{if } X_t \neq \hat{X}_t.
    \end{cases}
\end{equation}

\subsection{Problem Formulation}
This work aims to minimize the long-term average distortion subject to constraints on the normalized long-term average AoI of state $1$ and the average cost, by solving the following: 
\begin{subequations}\label{eq:main-opt-problem}
\begin{align}
	D^*  =\underset{\pi}{\text{min}}&  \lim_{T\rightarrow \infty}\frac{1}{T}\sum_{t=1}^{T}\mathbb{E}\left[\Delta(X_{t}, \hat{X}_{t})\right],\\
	\text{subject to}&\;  \lim_{T\rightarrow \infty}  \frac{1}{T} \sum_{t=1}^{T}  \mathbb{E} \left[a_{1}^d(t)\right] \leq \bar{A},\label{eq:avg-alarmstate-age}\\ 
	 &  \lim_{T\rightarrow \infty}\frac{1}{T}\sum_{t=1}^{T}\mathbb{E}\left[ c_1 I_1(t) + c_2 I_2(t) + c_3 I_3(t)  \right]\leq \bar{C},\label{eq:avg-cost}\\
     &I_1(t)+ I_2(t)+ I_3(t) \leq 1, \; \forall t\in \{1,2,\ldots\},
\end{align}
\end{subequations}
where $\bar{C}$ is the threshold on the average cost, $\bar{A}$ is the threshold on the average age of State $1$, and  $\pi$ is a policy, a specification of the decision rule to be used at each slot $t\in \{1,2\ldots\}$. The expectations are taken with respect to the chosen policy $\pi$, the stochastic process, $X_t$, and the channel state,  $h_t$. 

\begin{table}[t]
\vspace{.1cm}
\caption{Notations}
\centering 
\begin{tabular}{r c p{6cm} }
\toprule
$X_{t}$ & $\triangleq$ & The state of the process at time $t$\\
$\hat{X}_t$ & $\triangleq$ & The estimate of the state of $X_t$ at the destination\\
$h_t$ & $\triangleq$ & The state of the channel \\
$p_{i,j}^X$ & $\triangleq$ & Transition probability of the process from State $i$ to State $j$\\
$p_{i,j}^h$ & $\triangleq$ &  Transition probability of the channel from State $i$ to State $j$ \\  
$p_0^p$ & $\triangleq$ & The probability of successfully receiving the processed information when the channel is in bad state\\
$p_1^r$ & $\triangleq$ & The probability of successfully receiving the raw update when the channel is in good state, State $1$\\
$a_1^d(t)$ & $\triangleq$ & The age of information of State $1$ at the destination at time $t$ \\
$\bar{A}$ & $\triangleq$ & The bound on the average age of State $1$ \\
$\bar{C}$ & $\triangleq$ & The bound on the average cost \\
$c_1$ & $\triangleq$ & The cost incurred to sample and transmit the raw information\\
$c_2$ & $\triangleq$ & The cost incurred to sample, process and not transmit the information\\
$c_3$ & $\triangleq$ & The cost incurred to sample, process and transmit the processed information\\
$p_s$ & $\triangleq$ & The probability of successfully receiving either raw or processed information at the destination\\
$p_0$ & $\triangleq$ & The probability of choosing not to sample the process\\
$p_1$ & $\triangleq$ & The probability of choosing to sample the process and transmit the raw information\\
$p_2$ & $\triangleq$ & The probability of choosing to sample and process the information and not to transmit the information\\
$p_3$ & $\triangleq$ & The probability of choosing to sample, process and transmit the processed information \\
$p_{i,j}^{\Delta}$ & $\triangleq$ & The probability of transition of the distortion from State $i$ to State $j$ \\  
\bottomrule
\end{tabular}
\label{tab:TableOfNotationForMyResearch}
\end{table}

\section{Solution}
For solving the formulated problem \eqref{eq:main-opt-problem}, we propose two approaches, namely the Partially Observable Markov Decision Process (POMDP) framework that makes the decision in a sequential fashion, and the Stationary Randomized Policy (SRP) that makes the decision independent of the time slot. 
\subsection{Constrained POMDP Framework}
In each time slot, the source has to decide on sampling, processing, and transmission of the data without observing the state of the stochastic process $X_t$. POMDP framework can be utilized to solve \eqref{eq:main-opt-problem}. Since we also consider the long-term average age constraint \eqref{eq:avg-alarmstate-age} and the cost constraint \eqref{eq:avg-cost}, the formulated problem becomes a constrained POMDP problem.
A POMDP is defined by the tuple $(\mathcal{S}, \mathcal{U}, \mathcal{O}, \mathcal{P}, Z, c)$, where $\mathcal{S}$ represents the state space, $\mathcal{U}$ represents the action space, $\mathcal{O}$ represents the observation space, $\mathcal{P}$ represents the transition probabilities between the states conditioned on the action taken, $Z$ represents the observation probabilities, and $c$ represents the cost function. The components of POMDP for the formulated problem \eqref{eq:main-opt-problem} are explained as follows:

\subsubsection{State Space, $\mathcal{S}$}
We consider the state in slot $t$ to be, $\boldsymbol{s}_t = \left (X_t, \hat{X}_{t-1}, h_t, a_1^d(t) \right)$, where $X_t$ is the state of the process at the source in slot $t$, $\hat{X}_{t-1}$ is the estimate of the state $X_{t-1}$ at the destination, $h_t$ is the channel state in slot $t$ and $a_1^d(t)$ is the instantaneous AoI of State $1$ at the destination.
\subsubsection{Action Space, $\mathcal{U}$} 
The action in slot $t$ is represented as $u_t$ and the possible action in a slot is one of the following, (i) do not sample the stochastic process $X_t$, $u_t=0$, (ii) sample and transmit the raw information, $u_t=1$, (iii) sample, process, and decide not to transmit the information, $u_t=2$ and (iv) sample, process, and transmit the processed information, $u_t=3$.

\subsubsection{State Transition Probabilities, $\mathcal{P}$}
The transition probability from state $s_t=s$ to $s_{t+1}=s'$ under action $u_t=u$ in slot $t$ is represented by $\mathcal{P}(s'|s,u) = \mathbb{P}\{s_{t+1}=s'|s_t=s, u_t=u\}$ and can be obtained for different combinations of state and action pairs based on the evolution of the stochastic process at the source,    channel and the AoI of State $1$.

\subsubsection{Observation Space, $\mathcal{O}$}
The state $s_t$ in slot $t$ is partially observable since the decision maker cannot observe $X_t$. Based on the successful reception of the data and whether the received information was raw or processed, the possible observation $o_t$ in slot $t$ are that the destination receives (i)   no update, $o_t=0$, (ii)  the raw information, $o_t=1$ or (iii)  the processed information, $o_t=2$. 

\subsubsection{Observation Probabilities, $Z$}
The observation probability is given by $Z_{s'}(o,u) = \mathbb{P}\{o_t = o | s_{t+1} = s', u_t = u\}$ and represents the conditional probability of observing $o$ when action $u$ is taken and next state is $s'$. We can obtain it based on the different combinations of the state of stochastic process, channel state, AoI of State $1$ and actions. 

\subsubsection{Belief State Formulation}
To obtain the solution for a POMDP problem, we use the belief state formulation and obtain the continuous state MDP where the belief states are considered as the states. The belief state in a slot $t+1$ refers to the distribution over the state space $S$ and depends on the belief state in slot $t$, the action $u_t$, and the observation $o_t$. The belief about a state $s'$ is updated in every slot, and for a state $s'$, the belief update is given by:
\begin{equation}\label{eq:belief-update}
    b_o^u(s') = \frac{Z_{s'}(o,u) \sum_{s \in \mathcal{S}} \mathcal{P}(s'|s,u) b(s)}{\sum_{s' \in \mathcal{S}}  Z_{s'}(o,u) \sum_{s \in \mathcal{S}} \mathcal{P}(s'|s,u) b(s)},
\end{equation}
where $b(s)$ is the belief associated with state $s$, $b_o^u(s')$ is the updated belief for state $s'$ when action $u$ is taken and the observation $o$ is seen.
\subsubsection{Instantaneous Cost, $c$}
The instantaneous cost in state $s_t$ and action $u_t$ depends on the  (i)   distortion, (ii)  age of State $1$, and (iii)  cost for sampling, processing, and transmission. \\

In the above, we successfully cast the problem as a constrained POMDP, but solving it computationally inexpensively is challenging. We will consider this as part of future work. In the following, we obtain a low-complexity sub-optimal stationary randomized policy to solve \eqref{eq:main-opt-problem}. 

\subsection{Stationary Randomized Policy (SRP)}
A sub-optimal solution for the formulated problem \eqref{eq:main-opt-problem} can be obtained using an SRP. The SRP specifies the probability of different actions being chosen independently of slot index $t$ and the ACK. Let $p_0$ be the probability of choosing not to sample the process, $p_1$ be the probability of choosing to sample the process and transmit the raw information, $p_2$ be the probability of choosing to sample and process the information and not to transmit the information, and $p_3$ be the probability of choosing to sample, process and transmit the processed information. To obtain the optimal SRP, we reformulate the original problem \eqref{eq:main-opt-problem} using the expressions for the expected distortion and expected AoI of state $1$, which we derive next. 

\subsubsection{Derivation of Expected Distortion and AoI under SRP}
\begin{theorem}\label{thm:Dist_AoI}
The expected distortion and AoI under the SRP are given by 
    \begin{equation}\label{eq:distortion-srp}
     \mathbb{E}\left[\Delta(X_t, \hat{X}_t)\right] = 
     \frac{ 2 (1-p_s) p_{1,0}^X p_{0,1}^X}{(p_{0,1}^X + p_{1,0}^X) \left(1 + (1-p_s) \left(p_{1,0}^X - p_{0,0}^X\right)\right)},
\end{equation}
and 
\begin{equation}\label{eq:AoIExpression}
    \mathbb{E}[a_1^d(t)] = \frac{1}{p_s} \left(1-p_s\mathbb{P}\{\hat{X}=0\} \right),
\end{equation}
respectively, where 
$p_s  = p_1 p_1^r \mathbb{P} \{ h = 1\} + p_3 \left( \mathbb{P} \{ h = 1\} + p_0^p \mathbb{P} \{ h = 0\}  \right)$,
is the probability that the destination receives the information successfully in a slot, and $\mathbb{P}\{\hat{X}=0\}$ is the steady state probability of the estimate of the stochastic process being in State $0$. 
\end{theorem}
\begin{proof}
\begin{figure}[t]
\vspace{0.08cm}
    \centering
    \begin{tikzpicture}[scale=0.45, transform shape, ->, >=stealth', shorten >=1pt, auto, node distance=5cm, semithick]
  \tikzstyle{state}=[circle, draw, minimum size=0.5cm, align=center]
  
  \node[state] (A) {State $0$ \\ $X,\hat{X} = (0,0)$};
  \node[state, right of=A, below of=A] (B) {State $1$ \\ $X,\hat{X}= (0,1)$};
  \node[state, left of=B, below of=A] (C) {State $2$ \\ $X,\hat{X}= (1,0)$};
  \node[state, below of=A,  node distance=10cm] (D) {State $3$ \\ $X,\hat{X}= (1,1)$};
  
  
  \draw (A) edge[loop right, out=135, in=45, distance=1.8cm] node[above] {$p^{X}_{0,0}$} (A);
  \draw (A) edge[bend right] node[above, rotate=90] {$p^{X}_{0,1} p_s$} (D);
  \draw(A) edge[bend right] node[above, rotate=50] {$p^{X}_{0,1} (1-p_s)$} (C);
  
  \draw (B) edge[loop right, out=45, in=-45, distance=1.8cm] node[below, rotate=90]  {$p^{X}_{0,0} (1-p_s)$} (B);
  \draw(B) edge[bend right] node[above, rotate=-50]  {$p^{X}_{0,0} p_s$} (A);
  \draw(B) edge[bend right=0] node[above, rotate=60]  {$p^{X}_{0,1}$} (D);
  
  \draw (C) edge[loop right, out=145, in=-135, distance=1.8cm] node[above,rotate=90]  {$p^{X}_{1,1} (1-p_s)$} (C);
  \draw(C) edge[bend right] node[below, rotate=-50] {$p^{X}_{1,1} p_s$} (D);
  \draw (C) edge[bend right=0] node[above, rotate=60] {$p^{X}_{1,0}$} (A);

  \draw (D) edge[loop right, out=-135, in=-45, distance=1.8cm] node[below]  {$p^{X}_{1,1}$} (D);
  \draw(D)  edge[bend right] node[above, rotate=-85] {$p^{X}_{1,0} p_s$} (A);
  \draw(D) edge[bend right] node[below, rotate=45] {$p^{X}_{1,0} (1-p_s)$} (B);
\end{tikzpicture}
     \caption{The evolution of two-dimensional DTMC with $(X, \hat{X})$ as the state.
     }
    \label{fig:MC-distortion}
\end{figure}
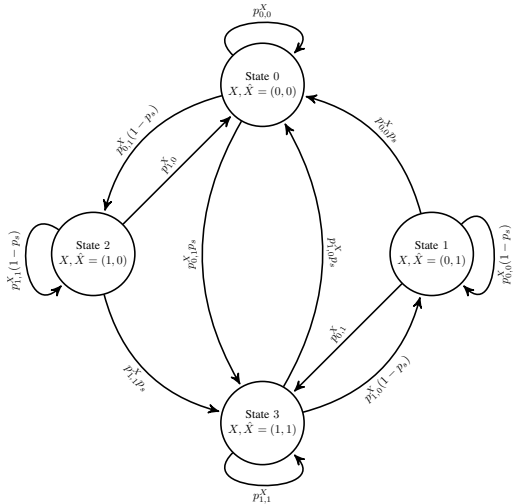
We first consider derivation of the expected distortion, which under the SRP is given by
\begin{align}\label{eq:distortion-sum}
    \mathbb{E}&\left[\Delta(X_t, \hat{X}_t)\right] = \mathbb{P}\{X_t \neq \hat{X}_t\}\nonumber\\
   & = \mathbb{P}\{X_t=0, \hat{X}_t=1\} +  \mathbb{P}\{X_t=1, \hat{X}_t=0\}.
\end{align}
Now, using $p_s$ and the fact that the steady-state probabilities of the channel states are $\mathbb{P} \{ h = 0\}  = {p_{1,0}^h}/({p_{0,1}^h + p_{1,0}^h})$ and $ 
\mathbb{P} \{ h = 1\}  = {p_{0,1}^h}/{(p_{0,1}^h + p_{1,0}^h)}$,  
we derive \eqref{eq:distortion-srp} in the following. 
We consider the two-dimensional Discrete-Time Markov Chain (DTMC) with the state of the stochastic process and the estimate of the stochastic process as the state, that is, $(X, \hat{X})$. The probability transition matrix $P$ for the considered Markov Chain (MC) illustrated in Fig. \ref{fig:MC-distortion} is given by,
\begin{equation*}
\begin{bmatrix}
p_{0,0}^X & 0 &  p_{0,1}^X (1-p_s) & p_{0,1}^X p_s\\
p_{0,0}^X p_s &  p_{0,0}^X (1-p_s) & 0 &  p_{0,1}^X \\
p_{1,0}^X & 0 & p_{1,1}^X (1-p_s) &  p_{1,1}^X p_s \\
p_{1,0}^X p_s &  p_{1,0}^X (1-p_s) & 0 &  p_{1,1}^X
\end{bmatrix}
\end{equation*}
We obtain the steady-state probabilities for each state of the MC by solving $\boldsymbol{\pi} P = \boldsymbol{\pi}$ where $\boldsymbol{\pi}  =[\pi_0, \pi_1, \pi_2, \pi_3]$ are the steady-state probabilities for the $4$-state MC, where 
\begin{align*}
    \pi_0 &= \frac{ \left(1 - p_{1,1}^X (1-p_s)\right) p_{1,0}^X}{(p_{0,1}^X + p_{1,0}^X) \left(1 + (1-p_s) \left(p_{1,0}^X - p_{0,0}^X\right)\right)}, \\
    \pi_1 &= \frac{ (1-p_s) p_{1,0}^X p_{0,1}^X}{(p_{0,1}^X + p_{1,0}^X) \left(1 + (1-p_s) \left(p_{1,0}^X - p_{0,0}^X\right)\right)}, \\
    \pi_2 &= \frac{ (1-p_s) p_{1,0}^X p_{0,1}^X}{(p_{0,1}^X + p_{1,0}^X) \left(1 + (1-p_s) \left(p_{1,0}^X - p_{0,0}^X\right)\right)}, 
    \end{align*}
 \begin{align*}
    \pi_3 &= \frac{ \left(1-(1-p_s) p_{0,0}^X \right) p_{0,1}^X}{(p_{0,1}^X + p_{1,0}^X) \left(1 + (1-p_s) \left(p_{1,0}^X - p_{0,0}^X\right)\right)}.
        \end{align*}
The result in  \eqref{eq:distortion-srp} follows using  \eqref{eq:distortion-sum} and noting that $\mathbb{E}[\Delta(X_t, \hat{X}_t)] = \pi_1 + \pi_3$.

\begin{figure*}[t]
	\centering
	\resizebox{0.9\textwidth}{!}{%
    \begin{tikzpicture}[scale=0.45, transform shape, ->, >=stealth', shorten >=1pt, auto, node distance=5cm, semithick]
  \tikzstyle{state}=[circle, draw, minimum size=1.5cm, align=center]
  
  \node[state] (0) {$0$};
  \node[state, right of=0] (1) {$1$};
  \node[state, right of=1] (2) {$2$};
  \node[state, right of=2] (n1) {$(n-1)$};
  \node[state, right of=n1] (n) {$n$};
  
  
  \draw (0) edge[loop right, out=145, in=-135, distance=1.8cm] node[above, rotate=90] {$p_s\mathbb{P}\{\hat{X}=0\}$} (0);
  \draw(0)  edge[bend left=50] node[above] {$(1-p_s) + p_s \mathbb{P}\{\hat{X}=1\}$} (1.north);

  \draw (1) edge[loop right, out=145, in=-135, distance=1.8cm] node[above, rotate=90] {$p_s\mathbb{P}\{\hat{X}=1\}$} (1);
    \draw (1)  edge[bend left] node[above] {$(1-p_s)$} (2);
  \draw (1)  edge[bend left=60] node[below, pos = 0.25, anchor= 80] {$p_s\mathbb{P}\{\hat{X}=0\}$} (0.south);

   \draw (2)  edge[bend left=50] node[below] {$p_s\mathbb{P}\{\hat{X}=1\}$} (1.south);
   \draw (2)  edge[bend left=60] node[below, pos=0.35] {$p_s\mathbb{P}\{\hat{X}=0\}$} (0.south);

   \path[->] (2)  edge[bend left] node[above] {$(1-p_s)$} (12.3,0.2);
   
   \path (2) -- node[auto=false]{\ldots} (n1);

   \draw (n1)  edge[bend left=60] node[below, pos = 0.35, anchor=60] {$p_s\mathbb{P}\{\hat{X}=0\}$} (0.south);
   \draw (n1)  edge[bend left] node[above] {$(1-p_s)$} (n);
   \path (12.68,0.2)  edge[bend left] node[above] {$(1-p_s)$} (n1);

    \draw (n)  edge[bend left=60] node[below, pos = 0.35, anchor=60] {$p_s\mathbb{P}\{\hat{X}=0\}$} (0.south);
    \draw (n)  edge[bend left] node[below] {$p_s\mathbb{P}\{\hat{X}=1\}$} (n1);
    \path (n)  edge[bend left] node[above] {$(1-p_s)$} (22.3,0.2);

    \path (n) -- node[auto=false]{\ldots} (24,0);
 \end{tikzpicture}}
 \vspace{-0.5cm}
    \caption{The evolution of the one-dimensional DTMC with the age of State $1$, $a_1^d$ as the state.} 
    \label{fig:MC-age}
\end{figure*}
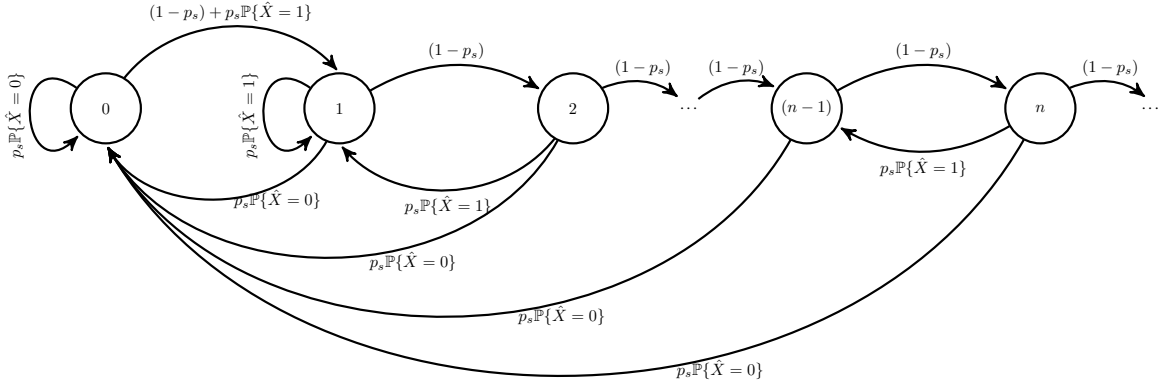
\begin{figure*}[!bp]
\hrule
\begin{align}\label{eq:lumpable-MC}
P &= 
\begin{blockarray}{cccccccc}
 \vspace{1ex}
& (0,0) & (0,1) & (1,0) & (1,1) & (2,0) & (2,1) & \dots \\
\begin{block}{c(ccccccc)}
 \vspace{2ex}
 (0,0) & p_s \mathbb{P}\{\hat{X}=0\}  & 0 & p_s \mathbb{P}\{\hat{X}=1\} + (1-p_s) p_{0,0}^{\Delta} & (1-p_s) p_{0,1}^{\Delta} & 0 & 0  & \dots\\
 \vspace{2ex}
 (0,1) & p_s\mathbb{P}\{\hat{X}=0\}  & 0 & p_s \mathbb{P}\{\hat{X}=1\} + (1-p_s) p_{1,0}^{\Delta} & (1-p_s) p_{1,1}^{\Delta} & 0 & 0  & \dots\\
  \vspace{2ex}
 (1,0) & p_s \mathbb{P}\{\hat{X}=0\}  & 0 & p_s \mathbb{P}\{\hat{X}=1\} & 0 & (1-p_s) p_{1,0}^{\Delta} & (1-p_s) p_{0,0}^{\Delta}  & \dots\\
  \vspace{2ex}
 (1,1) & p_s \mathbb{P}\{\hat{X}=0\}  & 0 & p_s \mathbb{P}\{\hat{X}=1\} & 0 & (1-p_s) p_{1,1}^{\Delta} & (1-p_s) p_{0,0}^{\Delta}  & \dots\\
  \vspace{2ex}
 \vdots & \vdots & \vdots & \vdots & \vdots & \vdots & \ddots \\
\end{block}
\end{blockarray}
\end{align}
\end{figure*}

%
The expected AoI can be obtained by computing its stationary distribution. However, the transition probability may depend on the current and subsequent distortion, so the evolution of AoI and distortion may be coupled. Hence, we consider the two-dimensional DTMC with the AoI of the State $1$ and the distortion as the state of the Markov chain, that is, $(a_1^d,\; |X-\hat{X}|)$, and the corresponding transition matrix $P$ given by \eqref{eq:lumpable-MC}. We first show that the DTMC is lumpable and obtain a one-dimensional DTMC with AoI as the state to obtain \eqref{eq:AoIExpression}.

Consider the partition of the two-dimensional state $\mathcal{S}$ of MC as $\mathcal{S} =\{ \mathcal{S}_0, \mathcal{S}_1, \dots \}$ where $\mathcal{S}_i= \{(i,0), (i,1)\}$. 
The transition matrix $P$ for the two-dimensional DTMC is lumpable with respect to the partition above, that is, $\mathbb{P}\{a_1^d(t+1) =j, \Delta(X_{t+1}, \hat{X}_{t+1}) = 0|a_1^d(t) =i, \Delta(X_t, \hat{X}_t) = 0\} + \mathbb{P}\{a_1^d(t+1) =j, \Delta(X_{t+1}, \hat{X}_{t+1}) = 1|a_1^d(t) =i, \Delta(X_t, \hat{X}_t) = 0\} = \mathbb{P}\{a_1^d(t+1) =j, \Delta(X_{t+1}, \hat{X}_{t+1} = 0|a_1^d(t) =i, \Delta(X_t, \hat{X}_t = 1\} + \mathbb{P}\{a_1^d(t+1) =j, \Delta(X_{t+1}, \hat{X}_{t+1} = 1|a_1^d(t) =i, \Delta(X_t, \hat{X}_t) = 1\}, \; \forall \; \mathcal{S}_i, \mathcal{S}_j$. Therefore, the two-dimensional DTMC can be reduced to the one-dimensional DTMC with the age of the State $1$, that is, $a_1^d$ as the state as shown in Fig. \ref{fig:MC-age}\cite{Lumpable-MC}. Now, from the definition of AoI of State $1$, i.e., \eqref{eq:age_evolution}, the expected AoI of State $1$ is given by $ \mathbb{E}[a_1^d(t)] = \sum_{i=1}^\infty i \cdot p_i$, where $p_i$ is the probability that age is equal to $i$. The DTMC for the AoI of State $1$ is illustrated in Fig. \ref{fig:MC-age}. We can show that $p_0 = p_s \mathbb{P}\{\hat{X}=0\}$, $p_1 =p_s \left(1-p_s\mathbb{P}\{\hat{X}=0\} \right)$ and $p_n = (1-p_s)^{n-1} p_1, \; \forall n > 1$. Therefore, the expected AoI of State $1$ is given by \eqref{eq:AoIExpression}. 
\end{proof}





\subsubsection{Reformulation of  \eqref{eq:main-opt-problem} under SRP} Under SRP, using Theorem~\ref{thm:Dist_AoI}, 
\eqref{eq:main-opt-problem} can be reformulated as the following problem: 
\begin{subequations}\label{eq:srp-opt-problem}
\begin{align}
	\underset{p_0, p_1, p_2, p_3\geq 0}{\text{minimize }} &\; \;\frac{ 2 (1-p_s) p_{1,0}^X p_{0,1}^X}{(p_{0,1}^X + p_{1,0}^X) \left(1 + (1-p_s) \left(p_{1,0}^X - p_{0,0}^X\right)\right)}, \label{eq:srp-distortion}\\ 
	\text{subject to}&\;\; \frac{1}{p_s} \left(1 - p_s \mathbb{P}\{\hat{X} = 0\}\right) \leq \bar{A},  \label{eq:srp-age}\\
  & c_1 p_1 + c_2 p_2 + c_3 p_3\leq \bar{C}, \label{eq:srp-cost} \\
  & \;\; p_0 + p_1 + p_2 + p_3 = 1, \label{eq:srp-sum-prob} 
\end{align}
\end{subequations}
where we recall that $p_s  = p_1 p_1^r\mathbb{P}\{h=1\} + p_3 \left( \mathbb{P}\{h=1\} + p_0^p \mathbb{P}\{h=0\} \right)$ is the probability that the destination receives the information successfully in a slot. We can observe that by 
setting $p_0 = 1 - p_1 - p_2 - p_3$ the equality constraint \eqref{eq:srp-sum-prob} can be rewritten as $p_1 + p_2 + p_3 \leq 1$. Furthermore, since SRP takes the decision independently of the ACK that the source receives, deciding to sample, process, and not transmit the information is not beneficial. Hence, we consider the optimal value for $p_2$ to be equal to $0$ and \eqref{eq:srp-opt-problem} can be transformed to the following: 
\begin{subequations}\label{eq:srp2-opt-problem}
\begin{align}
	\underset{0 \leq p_1,\leq p_3 \leq 1}{\text{minimize }} &\; \;\frac{ 2 (1-p_s) p_{1,0}^X p_{0,1}^X}{(p_{0,1}^X + p_{1,0}^X) \left(1 + (1-p_s) \left(p_{1,0}^X - p_{0,0}^X\right)\right)}, \\ 
	\text{subject to}&\;\; \frac{1}{p_s} \left(1 - p_s \mathbb{P}\{\hat{X} = 0\}\right) \leq \bar{A},\label{eq:ageC} \\
  & \;\;c_1 p_1 + c_3 p_3\leq \bar{C}, \;\; p_1 + p_3 \leq 1. \label{eq:srp2-sum-prob}
\end{align}
\end{subequations}
Since $p_s$ is an affine function of $p_1$ and $p_3$, the objective function of the above problem is a ratio of affine functions with a non-zero denominator. Furthermore, \eqref{eq:ageC} can be converted into an affine constraint, and \eqref{eq:srp2-sum-prob} represents an affine constraint. Thus, \eqref{eq:srp2-opt-problem} can be expressed as a linear fractional program, which can be further simplified to a linear program if a feasible solution exists. Hence, the above optimization problem can be solved using standard optimization techniques.

 \begin{figure*}[t]
    \centering
    \begin{subfigure}[t]{0.48\textwidth}
        \resizebox{\linewidth}{!}{\input{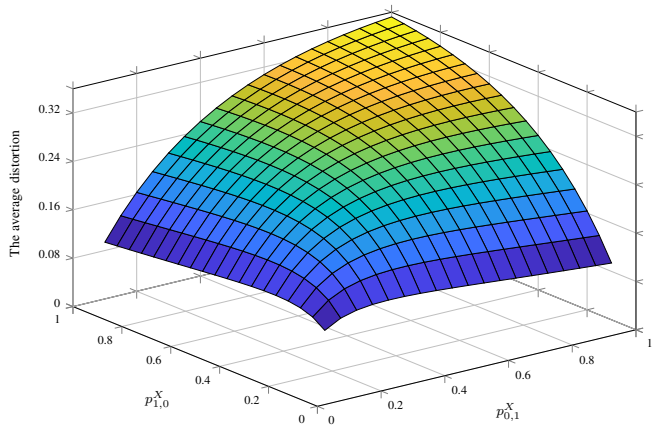}}
        \caption{The variation of the optimal average distortion with respect to $p_{0,1}^X$ and $p_{1,0}^X$, when $p_{0,1}^h = 0.2, p_{1,0}^h=0.3, p_1^r=0.9, p_0^p = 0.6, \bar{A}=3, c_1 = 1.2, c_2 = 0.8, c_3 = 1$ and  $\bar{C} = 0.8$.}
        \label{fig:Distortion-state}
    \end{subfigure}
    \hfill
    \begin{subfigure}[t]{0.48\textwidth}
        \resizebox{\linewidth}{!}{\input{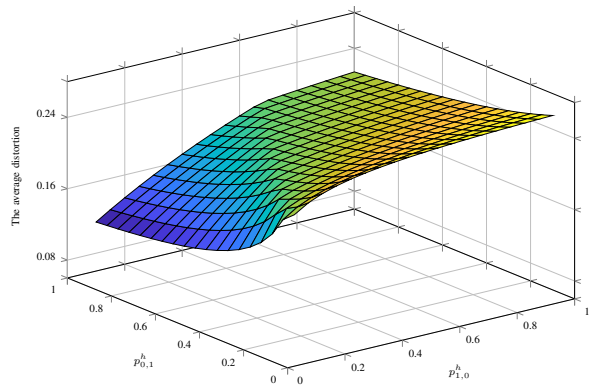}}
        \caption{The variation of the optimal average distortion with respect to $p_{0,1}^h$ and $p_{1,0}^h$, when $p_{0,1}^X = 0.2, p_{1,0}^X = 0.3, p_1^r=0.9, p_0^p = 0.6, \bar{A}=3, c_1 = 0.8, c_2 = 0.7, c_3 = 1.25$ and  $\bar{C} = 0.6$.}
        \label{fig:Distortion-channel}
    \end{subfigure}

    \caption{The variation of the optimal average distortion with respect to transition probabilities of source and channel states.}
    \label{fig:overall_X_h}
\end{figure*}

\section{Simulation Results}
In this section, we study the impact of the system parameters on the average distortion obtained by solving \eqref{eq:srp2-opt-problem}. 

Recall that the states of the process being monitored and the channel over which the samples are transmitted evolve as separate two-state Markov chains. We study the impact of their transition probabilities on the average distortion in Fig.~\ref{fig:overall_X_h}. In Fig. \ref{fig:Distortion-state}, we present the variation of the average distortion with the source state's transition probabilities, $p_{0,1}^X$ and $p_{1,0}^X$. When the values of $p_{0,1}^X$ and $p_{1,0}^X$ are close to $1$, the state of the process changes frequently from state $0$ to state $1$ and vice-versa, due to which the information gets outdated quickly. This results in a higher distortion when the latest received state is used to estimate the current state. Frequent transmissions are needed to reduce distortion, which may not be possible due to the average cost constraint, resulting in higher distortion. However, when $p_{0,1}^X$ or $p_{1,0}^X$ are close to $1$ or $0$, the states change less frequently, and, even without transmission, the estimated state at the destination and the actual state at the source can coincide, causing lower distortion. 
Similarly, in Fig. \ref{fig:Distortion-channel}, we present the variation of the average distortion with the transition probabilities of the channel states, namely $p_{0,1}^h$ and $p_{1,0}^h$. When $p_{0,1}^h$ is high, and $p_{1,0}^h$ is low, the channel remains mainly in a good state; hence the probability of successfully receiving a transmitted update is higher, leading to lower distortion. On the other hand, when $p_{0,1}^h$ is low, and $p_{1,0}^h$ is high, the channel remains mainly in a bad state, and the probability of successful reception of a transmitted update is lower, causing higher distortion. When both $p_{0,1}^h$ and $p_{1,0}^h$ are high or low, the channel remains in good and bad states for a similar duration of time, resulting in intermediate distortions. 

These observations are further expanded in Fig. \ref{fig:Distortion-Age}, where we present the optimal average distortion achievable vs. normalized average age. We obtain the plot by considering a weighted sum of average distortion and the normalized average age as the objective function, which we minimize subject to an average cost constraint. We observe that when the state of the source changes rapidly, that is, when $p^X_{0,1}$ and $p^X_{1,0}$ are high, the distortion tends to be higher, which is further compounded if the probability of the channel being in a bad state is high, that is when $p^h_{0,1}$ is low and $p^h_{1,0}$ is high. Similarly, when the probability of the channel being in a good state is high, i.e., when $p^h_{0,1}$ is high and $p^h_{1,0}$ is low, the age tends to be lower, which decreases further if the source state changes slowly.  

Next, we study the impact of the costs of sampling and transmission and that of the bound on the average cost incurred in Fig.~\ref{fig:distortion_costs}. In Fig. \ref{fig:Distortion-cost}, we present the variation of the average distortion with respect to the overall cost for sampling and transmitting the raw information, $c_1$, and the overall cost for sampling, processing and transmitting the processed information, $c_3$. The selected parameters are such that the state of the source changes infrequently, necessitating only occasional delivery of status updates to keep the average distortion low. Moreover,  the channel remains in good and bad states for roughly equal amounts of time. Furthermore, if the raw sample is transmitted, it is successful with a high probability when the channel is in good state and unsuccessful when it is in bad state.
On the other hand, the processed sample is successful with probability $0.6$ when the channel is in the bad state and always successful when the channel is in the good state. From the figure, we observe that \eqref{eq:srp2-opt-problem} is infeasible for high values of $c_1$ and $c_3$ since transmitting frequently enough to meet the average AoI constraint may violate the average cost constraint, and transmitting infrequently enough to meet the average cost constraint may violate the average AoI constraint. However, when either $c_1$  or $c_2$ is lower, the distortion is relatively low, as one can either transmit the raw or processed samples while respecting the AoI and cost constraints. 
In Fig. \ref{fig:Distortion-Cbar}, we present the variation of the optimal average distortion with respect to the bound on the average cost, $\bar{C}$. As $\bar{C}$ increases, the average distortion decreases and eventually reaches a plateau. This is because an increase in $\bar{C}$ allows for more frequent transmissions resulting in lower distortion, and beyond a specific value of $\bar{C}$, we can sample, process and transmit in every slot, achieving lowest possible distortion, that cannot be reduced further even if $\bar{C}$ is increased. It is important to note that the distortion does not reduce to zero even for large values of $\bar{C}$. This is because the channel can be in a bad state with a non-zero probability, and the success probability of transmissions over the bad state, even after processing, is strictly less than one. Hence, although we can transmit in every slot when $\bar{C}$ is large, not all the transmissions are successful, leading to a non-zero distortion.  
We also note that for any value of $\bar{C}$, the average distortion depends on whether the transition of the state of the stochastic process at the source changes slowly or rapidly; if the source transitions slowly, the average distortion tends to be small and vice versa. 

\begin{figure*}[t]
    \centering
        \begin{subfigure}[t]{0.48\textwidth}
        \centering
        \resizebox{\linewidth}{!}{ 
            \input{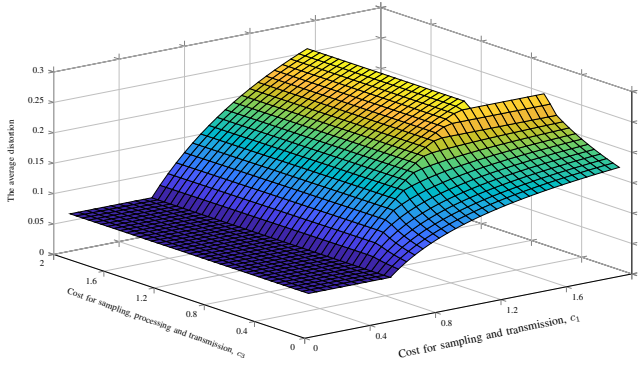}
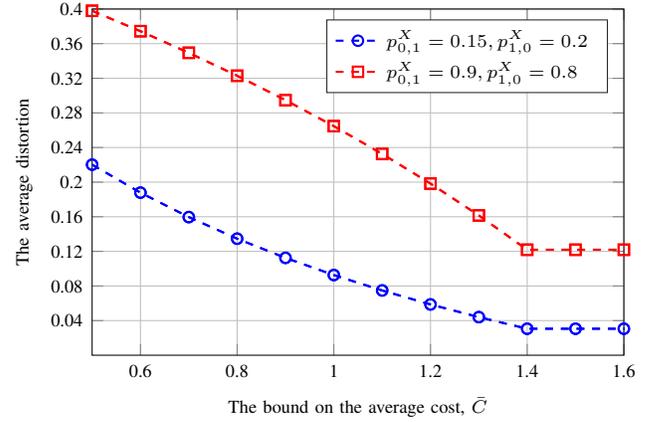
        }
        \caption{The variation of the optimal average distortion with respect to the cost incurred to sample and transmit the raw sample, $c_1$, and the cost incurred to sample, process and transmit the processed information, $c_3$,  when $p_{0,1}^X = 0.4, p_{1,0}^X=0.2, p_{0,1}^h = 0.2, p_{1,0}^h=0.3, p_1^r=0.9, p_0^p = 0.6, \bar{A}=6, c_2 = 0.8$ and  $\bar{C} = 0.8$.}
        \label{fig:Distortion-cost}
    \end{subfigure}
        \hfill
    \begin{subfigure}[t]{0.48\textwidth}
        \centering
        \resizebox{\linewidth}{!}{ 
\begin{tikzpicture}[font=\scriptsize]

\begin{axis}[%
width=2.821in,
height=1.838in,
at={(2.08in,0.858in)},
scale only axis,
xmin=0.5,
xmax=1.6,
xtick={0.2, 0.4, 0.6, 0.8, 1, 1.2, 1.4, 1.6},
xticklabels={0.2, 0.4, 0.6, 0.8, 1, 1.2, 1.4, 1.6},
xlabel style={font=\scriptsize,
	yshift=0.5ex,
	name=label},
xmajorgrids,
xlabel={The bound on the average cost, $\bar{C}$},
ymin=0,
ymax=0.4,
ytick={0.04, 0.08, 0.12, 0.16, 0.2, 0.24, 0.28, 0.32, 0.36, 0.4},
yticklabels={0.04, 0.08, 0.12, 0.16, 0.2, 0.24, 0.28, 0.32, 0.36, 0.4},
ylabel style={font=\scriptsize,
	yshift=-2ex,
	name=label},
ymajorgrids,
ylabel={The average distortion},
axis background/.style={fill=white},
title style={font=\bfseries},
title={},
legend columns=1, 
legend cell align={left}, 
legend style={legend cell align=left, align=left, legend pos = north east, draw=white!15!black}
]
\addplot[color=blue,dashed,thick, mark=o, mark options={solid, blue}, mark size = 2pt]
  table[row sep=crcr]{%
0.5	0.220183486238532\\
0.6	0.187866927592955\\
0.7	0.159596422838798\\
0.8	0.134656883901597\\
0.9	0.112492357856124\\
1	0.0926640926640927\\
1.1	0.0748211993398129\\
1.2	0.058679706601467\\
1.3	0.0440073345557593\\
1.4	0.0306122448979592\\
1.5	0.0306122448979592\\
1.6	0.0306122448979592\\
1.7	0.0306122448979592\\
1.8	0.0306122448979592\\
};
\addlegendentry{$p_{0,1}^X = 0.15, p_{1,0}^X=0.2$ }

\addplot[color=red, dashed,thick, mark=square, mark options={solid, red}, mark size = 2pt]
  table[row sep=crcr]{%
0.5	0.397947098302408\\
0.6	0.374390640233994\\
0.7	0.349426826207012\\
0.8	0.322925651198896\\
0.9	0.294740589125211\\
1	0.264705882352941\\
1.1	0.232633279483037\\
1.2	0.198308085776116\\
1.3	0.161484351106127\\
1.4	0.121878967414304\\
1.5	0.121878967414304\\
1.6	0.121878967414304\\
1.7	0.121878967414304\\
1.8	0.121878967414304\\
};
\addlegendentry{$p_{0,1}^X = 0.9, p_{1,0}^X=0.8$ }

\addplot [color=blue,dashed, thick, mark=o, mark options={solid, blue}, mark size = 2pt]
  table[row sep=crcr]{%
0.5	0.220183486238532\\
0.6	0.187866927592955\\
0.7	0.159596422838798\\
0.8	0.134656883901597\\
0.9	0.112492357856124\\
1	0.0926640926640927\\
1.1	0.0748211993398129\\
1.2	0.058679706601467\\
1.3	0.0440073345557593\\
1.4	0.0306122448979592\\
1.5	0.0306122448979592\\
1.6	0.0306122448979592\\
1.7	0.0306122448979592\\
1.8	0.0306122448979592\\
};

\addplot [color=red, dashed, thick, mark=square, mark options={solid, red}, mark size = 2pt]
  table[row sep=crcr]{%
0.5	0.397947098302408\\
0.6	0.374390640233994\\
0.7	0.349426826207012\\
0.8	0.322925651198896\\
0.9	0.294740589125211\\
1	0.264705882352941\\
1.1	0.232633279483037\\
1.2	0.198308085776116\\
1.3	0.161484351106127\\
1.4	0.121878967414304\\
1.5	0.121878967414304\\
1.6	0.121878967414304\\
1.7	0.121878967414304\\
1.8	0.121878967414304\\
};

\end{axis}
\end{tikzpicture}%
        }
        \caption{The variation of the optimal distortion with respect to the bound on the average cost $\bar{C}$, when parameters $p_{0,1}^h = 0.3, p_{1,0}^h=0.2, p_1^r=0.9, p_0^p = 0.6, \bar{A} = 3, c_1 = 1.2, c_2 = 0.8$ and $c_3 = 1.4$.}
        \label{fig:Distortion-Cbar}
    \end{subfigure}%
    \caption{The variation of the optimal average distortion with respect to the costs for different actions and the bound on the average cost.}
    \label{fig:distortion_costs}
\end{figure*}
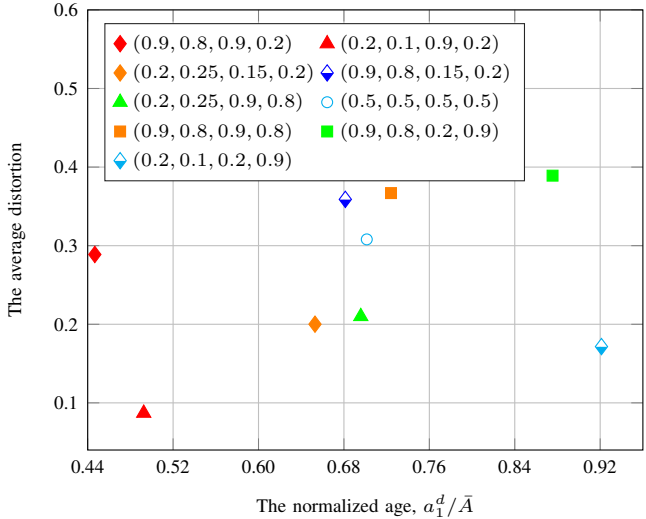
\begin{figure}[t]
    \centering
    \resizebox{0.48\textwidth}{!}{
\begin{tikzpicture}[scale=0.9, font=\scriptsize]

\begin{axis}[%
width=2.821in,
height=2.238in,
at={(2.08in,0.858in)},
scale only axis,
xmin=0.44,
xmax=0.96,
xtick={0.44, 0.52, 0.60, 0.68, 0.76, 0.84, 0.92},
xticklabels={0.44, 0.52, 0.60, 0.68, 0.76, 0.84, 0.92},
xlabel style={font=\scriptsize,
	yshift=0.5ex,
	name=label},
xmajorgrids,
xlabel={The normalized age, $a_1^d/\bar{A}$},
ymin=0.04,
ymax=0.6,
ytick={0, 0.1, 0.2, 0.3, 0.4, 0.5, 0.6, 0.7, 0.8, 0.9, 1},
yticklabels={0, 0.1, 0.2, 0.3, 0.4, 0.5, 0.6, 0.7, 0.8, 0.9, 1},
ylabel style={font=\scriptsize,
	yshift=-2ex,
	name=label},
ymajorgrids,
ylabel={The average distortion},
axis background/.style={fill=white},
title style={font=\bfseries},
title={},
legend columns=2, 
legend cell align={left}, 
legend style={legend cell align=left, align=left, legend pos = north west, draw=white!15!black}
]

\addplot[only marks, mark=diamond*, mark options={solid, red}, mark size = 3pt]
  table[row sep=crcr]{%
0.4467 0.2888\\
};
\addlegendentry{$(0.9, 0.8, 0.9, 0.2)$}

\addplot[only marks, mark=triangle*, mark options={solid, red}, mark size = 3pt]
  table[row sep=crcr]{%
0.4925 0.0869\\
};
\addlegendentry{$(0.2, 0.1, 0.9, 0.2)$}

\addplot[only marks, mark=diamond*, mark options={solid, orange}, mark size = 3pt]
  table[row sep=crcr]{%
0.6529	0.2001\\
};
\addlegendentry{$(0.2, 0.25, 0.15, 0.2)$}

\addplot[only marks, mark=halfdiamond*, mark options={solid, blue}, mark size = 3pt]
  table[row sep=crcr]{%
0.6813 0.3589\\
};
\addlegendentry{$(0.9, 0.8, 0.15, 0.2)$}

\addplot[only marks, mark=triangle*, mark options={solid, green}, mark size = 3pt]
  table[row sep=crcr]{%
0.6958 0.2099\\
};
\addlegendentry{$(0.2, 0.25, 0.9, 0.8)$}

\addplot[only marks, mark=o, mark options={solid, cyan}, mark size = 2pt]
  table[row sep=crcr]{%
0.7014 0.3080\\
};
\addlegendentry{$(0.5, 0.5, 0.5, 0.5)$}

\addplot[only marks, mark=square*, mark options={solid, orange}, mark size = 2pt]
  table[row sep=crcr]{%
0.7241 0.3669\\
};
\addlegendentry{$(0.9, 0.8, 0.9, 0.8)$}

\addplot[only marks, mark=square*, mark options={solid, green}, mark size = 2pt]
  table[row sep=crcr]{%
0.8754 0.3891\\
};
\addlegendentry{$(0.9, 0.8, 0.2, 0.9)$}

\addplot[only marks, mark=halfdiamond*, mark options={solid, cyan}, mark size = 3pt]
  table[row sep=crcr]{%
0.9212 0.1716\\
};
\addlegendentry{$(0.2, 0.1, 0.2, 0.9)$}

\end{axis}
\end{tikzpicture}%
}
    \caption{The optimal average distortion and the optimal normalized average age when $p_1^r=0.9, p_0^p=0.6, c_1=0.8, c_2=0.7, c_3=1.25, \bar{C} = 0.6$ and $\bar{A}=3$. The tuple values in the legend correspond to $(p_{0,1}^X, p_{1,0}^X, p_{0,1}^h, p_{1,0}^h)$ .}
    \label{fig:Distortion-Age}
\end{figure}
\balance
\section{Conclusion and Future Work}
In this work, we considered timely and accurate communication of samples from a source that monitors a stochastic process to a destination over a wireless ON/OFF channel. The samples had to be processed to determine the state of the stochastic process, which could be performed at the source before transmission or at the destination after reception. The goal was to minimize the distortion at the destination about the stochastic process at the source while maintaining the AoI at the destination below a threshold and satisfying a cost constraint at the source. We proposed a stationary randomized policy (SRP) as a solution and derived expected distortion, and provided a simple derivation of expected AoI using a lumpability argument on the two-dimensional DTMC formed with AoI and instantaneous distortion as states. We observed that the SRP does not account for ACK for making transmission decisions, so sampling, processing, and not transmitting is not beneficial. We conclude that the distortion tends to be higher when the source changes rapidly, and the channel is more likely to remain in a bad state. 
In addition, we have proposed and defined a constrained POMDP formulation for the problem, solving which numerically and comparing the performance with the SRP and other competitive policies  is considered part of future work.

\bibliographystyle{ieeetr}
\bibliography{references}

\begin{thebibliography}{10}

\bibitem{aoi-sanjit}
S.~Kaul, R.~Yates, and M.~Gruteser, ``{Real-time status: How often should one
  update?},'' in {\em Proc.IEEE INFOCOM}, pp.~2731--2735, 2012.

\bibitem{Book_Nikos}
A.~Kosta, N.~Pappas, and V.~Angelakis, {\em {Age of Information: A New Concept,
  Metric, and Tool}}.
\newblock Now Foundations and Trends, 2017.

\bibitem{Book_Yates}
R.~D. Yates {\em et~al.}, ``{Age of Information: An Introduction and Survey},''
  {\em IEEE J. Sel. Areas Commun.}, vol.~39, no.~5, pp.~1183--1210, 2021.

\bibitem{Pappas2023age}
N.~Pappas, M.~Abd-Elmagid, B.~Zhou, W.~Saad, and H.~Dhillon, {\em {Age of
  Information: Foundations and Applications}}.
\newblock Cambridge University Press, Feb. 2023.

\bibitem{KostaTCOM2020}
A.~Kosta, N.~Pappas, A.~Ephremides, and V.~Angelakis, ``{The Cost of Delay in
  Status Updates and Their Value: Non-Linear Ageing},'' {\em IEEE Trans. on
  Commun.}, vol.~68, no.~8, pp.~4905--4918, 2020.

\bibitem{KostaJSAC2021}
A.~Kosta, N.~Pappas, A.~Ephremides, and V.~Angelakis, ``{The Age of Information
  in a Discrete Time Queue: Stationary Distribution and Non-Linear Age Mean
  Analysis},'' {\em {IEEE} J. Sel. Areas Commun.}, vol.~39, no.~5,
  pp.~1352--1364, 2021.

\bibitem{AoII}
A.~Maatouk, S.~Kriouile, M.~Assaad, and A.~Ephremides, ``{The Age of Incorrect
  Information: A New Performance Metric for Status Updates},'' {\em IEEE/ACM
  Trans. Netw.}, vol.~28, no.~5, pp.~2215--2228, 2020.

\bibitem{AoIE}
B.~Joshi, R.~V. Bhat, B.~N. Bharath, and R.~Vaze, ``{Minimization of Age of
  Incorrect Estimates of Autoregressive Markov Processes},'' in {\em WiOpt},
  pp.~1--8, 2021.

\bibitem{AoS}
J.~Zhong, R.~D. Yates, and E.~Soljanin, ``{Two Freshness Metrics for Local
  Cache Refresh},'' in {\em IEEE ISIT}, pp.~1924--1928, 2018.

\bibitem{VAoI_Yates}
R.~D. Yates, ``{The Age of Gossip in Networks},'' in {\em IEEE ISIT},
  pp.~2984--2989, 2021.

\bibitem{VAoI_Ulukus}
B.~Buyukates, M.~Bastopcu, and S.~Ulukus, ``{Version Age of Information in
  Clustered Gossip Networks},'' {\em IEEE J. Sel. Areas Inf. Theory}, vol.~3,
  no.~1, pp.~85--97, 2022.

\bibitem{AoA2023}
A.~Nikkhah, A.~Ephremides, and N.~Pappas, ``{Age of Actuation in a Wireless
  Power Transfer System},'' in {\em IEEE INFOCOM - Age of Information
  Workshop}, May 2023.

\bibitem{Preprocess-Tony}
X.~Wang, M.~Fang, C.~Xu, H.~H. Yang, X.~Sun, X.~Chen, and T.~Q.~S. Quek,
  ``{When to Preprocess? Keeping Information Fresh for Computing-Enable
  Internet of Things},'' {\em IEEE Internet Things J.}, vol.~9, no.~6,
  pp.~4303--4317, 2022.

\bibitem{AoP}
R.~Li, Q.~Ma, J.~Gong, Z.~Zhou, and X.~Chen, ``{Age of Processing: Age-Driven
  Status Sampling and Processing Offloading for Edge-Computing-Enabled
  Real-Time IoT Applications},'' {\em IEEE Internet Things J.}, vol.~8, no.~19,
  pp.~14471--14484, 2021.

\bibitem{AoPI}
S.~Jayanth and R.~Bhat, ``{Age of Processed Information Minimization over
  Fading Multiple Access Channels},'' {\em IEEE Trans. Wirel. Commun.},
  pp.~1--1, 2022.

\bibitem{age-distortion-RVB}
G.~G. B, J.~S, and R.~V. Bhat, ``{Age of Information Minimization with Power
  and Distortion Constraints in Multiple Access Channels},'' in {\em WiOpt},
  pp.~1--7, 2021.

\bibitem{QoI_Rahul}
N.~Rajaraman, R.~Vaze, and G.~Reddy, ``{Not Just Age but Age and Quality of
  Information},'' {\em IEEE J. Sel. Areas Commun.}, vol.~39, no.~5,
  pp.~1325--1338, 2021.

\bibitem{NikosICAS21}
N.~Pappas and M.~Kountouris, ``{Goal-Oriented Communication For Real-Time
  Tracking In Autonomous Systems},'' in {\em IEEE ICAS}, 2021.

\bibitem{control_alarm}
G.~Stamatakis, N.~Pappas, and A.~Traganitis, ``{Control of Status Updates for
  Energy Harvesting Devices That Monitor Processes with Alarms},'' in {\em IEEE
  GC Wkshps}, pp.~1--6, 2019.

\bibitem{Lumpable-MC}
E.~Fountoulakis, T.~Charalambous, N.~Nomikos, A.~Ephremides, and N.~Pappas,
  ``{Information freshness and packet drop rate interplay in a two-user
  multi-access channel},'' {\em J. Commun. Netw.,}, vol.~24, no.~3,
  pp.~357--364, 2022.

\end{thebibliography}
\end{document}